\documentclass[11pt]{article}  

\usepackage{graphicx}
\usepackage[all]{xy}
\usepackage{latexsym}
\usepackage[english]{babel}
\usepackage{amsmath}
\usepackage{amsfonts}
\usepackage{amssymb}
\usepackage{makeidx}
\usepackage[utf8]{inputenc}
\usepackage{lmodern}
\usepackage{verbatim}
\usepackage{cite} 
\usepackage[pagewise]{lineno}
\usepackage{amsthm}
\usepackage{tikz-cd} 

\newtheorem{theorem}{Theorem}[section]

\newtheorem{lemma}[theorem]{Lemma}
\newtheorem{proposition}[theorem]{Proposition}

\theoremstyle{definition}
\newtheorem{definition}[theorem]{Definition}
\newtheorem{remark}[theorem]{Remark}
\newtheorem{example}[theorem]{Example}

\newcommand{\R}{\ensuremath{\mathbb{R}}}

\newcommand{\D}{\ensuremath{\mathcal{D}}}
\newcommand{\Es}{\ensuremath{\mathbb{S}}}

\newcommand{\M}{\ensuremath{\mathcal{M}}}

\DeclareMathOperator{\arcsinh}{arcsinh}

\usepackage[square, authoryear]{natbib}
\numberwithin{equation}{section}

\begin{document}
	\title{Radial kinetic nonholonomic trajectories are Riemannian geodesics!}
	\author{
		{\bf\large Alexandre Anahory Simoes}\hspace{2mm}
		\vspace{1mm}\\
		{\small  Instituto de Ciencias Matem\'aticas (CSIC-UAM-UC3M-UCM)} \\
		{\small C/Nicol\'as Cabrera 13-15, 28049 Madrid, Spain}\\
		{\it\small e-mail: \texttt{alexandre.anahory@icmat.es }}\\
		\vspace{2mm}\\
		{\bf\large Juan Carlos Marrero}\hspace{2mm}
		\vspace{1mm}\\
		{\it\small 	ULL-CSIC Geometr{\'\i}a Diferencial y Mec\'anica Geom\'etrica,}\\
		{\it\small  {Departamento de Matem\'aticas, Estad{\'\i}stica e I O, }}\\
		{\it\small  {Secci\'on de Matem\'aticas, Facultad de Ciencias}}\\
		{\it\small  Universidad de la Laguna, La Laguna, Tenerife, Canary Islands, Spain}\\
		{\it\small e-mail: \texttt{{jcmarrer@ull.edu.es}} }\\
		\vspace{2mm}\\
		{\bf\large David Martín de Diego}\hspace{2mm}
		\vspace{1mm}\\
		{\small  Instituto de Ciencias Matem\'aticas (CSIC-UAM-UC3M-UCM)} \\
		{\small C/Nicol\'as Cabrera 13-15, 28049 Madrid, Spain}\\
		{\it\small e-mail:  \texttt{david.martin@icmat.es} }\\		
	}

	\date{}
	
	\maketitle
	
	\vspace{0.5cm}
	\begin{abstract}
		Nonholonomic mechanics describes the motion of systems constrai\-ned by nonintegrable constraints. One of  its most remarkable properties is that the derivation of the nonholonomic equations is not variational in nature. {However, in} this paper, we prove (Theorem 1.1) that for kinetic nonholonomic {systems}, the solutions starting from a fixed  point $q$ are  true geodesics for a family of Riemannian metrics on the image submanifold ${\mathcal M}^{nh}_q$  of the nonholonomic exponential map. {This implies a surprising result: the kinetic nonholonomic trajectories with starting point $q$, for sufficiently small times, minimize length in ${\mathcal M}^{nh}_q$!}  
	\end{abstract}
	
	\let\thefootnote\relax\footnote{\noindent AMS {\it Mathematics Subject Classification ({2020})}. Primary 70G45; Secondary  53B20, 53C21,
		37J60, 70F25
	}

\section{Introduction}
Nonholonomic systems are mechanical systems that limit their admissible motions imposing constraints on the velocities (see \citep*{Neimark, Bloch, Cortes, Pacific} and references therein). Typically,  nonholonomic constraints are imposed by prescribing a nonintegrable distribution on the configuration space. This is, for instance, a typical situation in rolling motion without slipping. One of the aspects that makes nonholonomic dynamics very appealing is that the equations are derived using a non-variational principle, the Lagrange-d'Alembert principle, contrary to Hamilton’s principle in Lagrangian mechanics (see \citep*{LM94, CdLMdDM2003, AKN2006, Lewis2000} for a comparison with variational systems subjected to nonintegrable constraints). Nonetheless, it is known that physical systems subjected to nonholonomic constraints move according to this principle. Thus, the lack of a variational theory, makes it more important to understand the structure of the nonholonomic equations of motion and isolate all the underlying geometric objects which govern this kind of systems, which requires the use of several  differential geometry techniques, in particular, Riemannian geometry {(see, for instance,  \citep*{Synge28, vranceanu, Lewis98,AMM2, GaMa} and the references therein).} 

\medskip

\noindent {\bf Kinetic nonholonomic systems and Riemannian geometry.} In this paper, we will continue to deepen {the research program} of studying nonholonomic mechanics in terms of Riemannian geometry. In this sense, a fundamental concept in Riemannian geometry is, without any doubt, that of a geodesic. One of its main properties is that geodesics minimize the length among curves connecting nearby points. Conversely, any curve minimizing length is necessarily a geodesic. To prove this key result in Riemannian geometry, it is  necessary to introduce different concepts and results as, for instance, the notion of geodesic flow,  exponential map, Gauss' lemma, among others (see \citep*{LeeR,O'Neill, docarmo}). From another point of view, Riemannian geodesics are the trajectories of a Lagrangian mechanical system of kinetic type (\citep*{AM78}). Formally, if $(Q, g)$ is a Riemannian manifold, then the geodesics are precisely the solutions of a mechanical system defined by the kinetic Lagrangian
\[
L_{g}(v_q)=\frac{1}{2}g(q)(v_q, v_q)\; ,\qquad v_q\in T_q Q.
\]
In other words, if we denote by $\nabla^g$ the Levi-Civita connection associated to $g$, then the {solutions of the} Euler-Lagrange equations for $L$ are exactly the curves $c$ which satisfy the geodesic equations, that is, 
\begin{equation}\label{geoeq}
\nabla^g_{\dot{c}(t)}\dot{c}(t)=0.
\end{equation}
However, if we introduce a distribution $\D$ on $Q$ and impose the condition that curves $c$ must have their tangent vectors lying in $\D$, that is, $\dot{c}(t)\in \D_{c(t)}$ for all $t$, the  picture becomes much more complex. In the sequel, we will denote by a triple $(Q, g, \D)$, a kinetic nonholonomic system. The Lagrange-d'Alembert principle states that now the solutions are those curves $c$ which satisfy the following equation (see Section \ref{section2}):
\begin{equation}\label{cnh}
	\nabla^g_{\dot{c}(t)}\dot{c}(t)\in \D^\bot_{c(t)}\; ,\qquad   \dot{c}(t)\in \D_{c(t)}
\end{equation}
where  $\D^{\bot}$  is the $g$-orthogonal complement to $\D$. Equivalently, we can describe the nonholonomic {trajectories} as the geodesics of an affine connection $\nabla^{nh}$ on $Q$ with initial condition satisfying the  nonholonomic constraints (see Equation \ref{nhconnection}). Of course, analogous to what happens with unconstrained geodesics, a curve $c$ is a solution of \eqref{cnh} if and only if it is the solution of Lagrange-d'Alembert equations for the Lagrangian $L_{g}$. Only in very exceptional cases the connection $\nabla^{nh}$ is the Levi-Civita connection for a Riemannian metric, in particular, this would imply that the distribution was integrable (see \citep*{Lewis98}).

\medskip

\noindent {\bf Subriemannian geometry and kinetic nonholonomic systems.} 
In general,  the equations of motion of nonholonomic systems derived using Lagrange-d'Alembert principle are not described using variational principles (see \citep*{CdLMdDM2003,Lewis2000}) and, in particular,  as the geodesic equations for a Riemannian metric. 
This ``non-variational" derivation of the equations is the reason for the different qualitative behaviour of nonholonomic mechanics in comparison with standard variational problems (non-preservation, in general, of a symplectic form, neither of Poisson structures, non validity of Noether's theorem...). This is why it is traditionally thought that nonholonomic  mechanics is a very distinguished case that does not share many fundamental properties of the remaining  mechanical systems derived using purely variational techniques (see \citep*{BKMM96,Bloch,Neimark,cendra}). In fact, given a triple $(Q, g, \D)$ we have two main different geometries associated to it
\begin{enumerate}
	\item the {\bf subriemannian geometry}, related with the purely variational procedure. In this case, one is  interested  in {\bf shortest paths} where we measure distance restricting to curves satisfying the constraints  $\dot{c}(t)\in \D_{c(t)}$ (see \citep*{Montgomery}).
	\item the {\bf nonholonomic geometry} described by the solutions of nonholonomic geodesic equations (\ref{cnh}). In this case  one is interested in {\bf straightest paths} following the  Gauss’ least constraint principle (\citep*{Hertz, Neimark}).
	\end{enumerate}
Thus, solutions of (\ref{cnh}) are geodesics of the nonholonomic connection $\nabla^{nh}$ but they are not usually related with the problem of finding the curve that minimizes length between two points. 

\medskip

\noindent {\bf The main result of the paper.}
{Despite our previous comments, in} this article we present the  surprising result that we can characterize  {\bf radial nonholonomic geodesics, i.e.  nonholonomic solutions starting from a given point $q\in Q$, as true Riemannian  geodesics} for  a family of Riemannian  metrics $g_{q}^{nh}$ defined in the image  ${\mathcal M}^{nh}_{q}$ of the  nonholonomic exponential map {at $q$}. In other words, Theorem \ref{main-theorem} { below, the main result of this paper,} shows that {\bf radial kinetic nonholonomic trajectories are length minimizing} in a specified Riemannian manifold, that is, nonholonomic trajectories minimize the functional
\begin{equation*}
	L^{nh}_{q}(c)=\int_{0}^{1}\|\dot{c} \|_{g^{nh}_{q}} \ dt
\end{equation*}
among all curves $c:[0,1]\rightarrow \M^{nh}_{q}$ with fixed endpoints and starting at $q$. As a consequence,  they are Riemannian geodesics. Perhaps more importantly,  Theorem \ref{main-theorem} opens the door to new developments in nonholonomic mechanics using  Riemannian  geometry techniques (Jacobi fields, {global minimizing properties of nonholonomic trajectories},  construction of variational integrators for nonholonomic mechanics, Hamiltonization or {Lagrangianization} of nonholonomic systems...). In Section \ref{conclusions} we explore with detail some of these {new} possible lines of research in nonholonomic mechanics.

To formulate Theorem  \ref{main-theorem}  we need the following canonical identification. 
For an open subset ${\mathcal U}$ of a real vector space $E$ and for $v \in {\mathcal U}$, we will use the canonical linear identification $^{\bf v}_{v}$ between $E$ and $T_{v}{\mathcal U}$ given by
\begin{equation}\label{vertical-lift-v}
w^{\bf v}_{v} = \frac{d}{dt}_{|t = 0}(v + tw), \; \; \mbox{ for } w \in E.
\end{equation}

\begin{theorem}\label{main-theorem}
	Let $(Q, g,\D)$ be a kinetic nonholonomic system and $q$ a fixed point in $Q$.Then:
	\begin{enumerate}
		\item[i)]
		There exists a submanifold ${\mathcal M}^{nh}_{q}$ of $Q$, with  $q\in {\mathcal M}^{nh}_{q}$, and a diffeomorphism
		$\emph{exp}^{nh}_{q}: {\mathcal U}_0 \subseteq \D_{q} \to {\mathcal M}^{nh}_{q} \subseteq Q$,
		{where $\, {\mathcal U}_0$ is a starshaped open subset of $\D_{q}$ about $0_{q}\in {\mathcal U}_{0}$} and $\emph{exp}^{nh}_{q}(0_{q}) = q$. The map $\emph{exp}^{nh}_{q}$ is the nonholonomic exponential map at $q$. Moreover, we have that:
		\begin{enumerate}
		\item
		Under the canonical linear identification between $\D_{q}$ and $T_{0_{q}}{\mathcal U}_0$, the linear monomorphism
		\[
		T_{0_{q}}\emph{exp}^{nh}_{q}: T_{0_{q}}{\mathcal U}_0 \simeq \D_{q} \to T_{q}Q
		\]
		is just the canonical inclusion of $\D_{q}$ in $T_{q}Q$.
		\item
		For every $v_{q} \in {\mathcal U}_{0}$,
		\begin{equation}\label{radial:standard:form}
			\emph{exp}^{nh}_{q}(tv_{q}) = c_{v_{q}}(t), \; \; t \in [0, 1],
		\end{equation}
		with $c_{v_{q}}: [0, 1] \to {\mathcal M}^{nh}_{q} \subseteq Q$ the (unique) nonholonomic trajectory satisfying
		$c_{v_{q}}(0) = q, \dot{c}_{v_{q}}(0) = v_{q}$.
		\end{enumerate}   
		\item[ii)] All the radial kinetic nonholonomic trajectories departing from the fixed point $q\in Q$ are {homothetic reparametrizations} of nonholonomic trajectories {given by equation \eqref{radial:standard:form}}. In addition, they are minimizing geodesics for a Riemannian metric $g_{q}^{nh}$ on ${\mathcal M}^{nh}_{q}$ if and if only if the Riemannian metric ${\mathcal G}_0 = (\emph{exp}^{nh}_{q})^*(g_{q}^{nh})$ on ${\mathcal U}_0$ satisfies the Gauss condition, that is,
		\[
		{\mathcal G}_{0}(v_{q})(v_{q}, w_{q}) = {\mathcal G}_{0}(0_{q})(v_{q}, w_{q}), \; \; \mbox{ for } v_{q} \in {\mathcal U}_0 \mbox{ and } w_{q} \in D_{q}.
		\]
		\item[iii)] Such Riemannian metrics on ${\mathcal M}^{nh}_{q}$ always exist and if $g^{nh}_q$ is one of them then the Riemannian exponential associated with $g^{nh}_q$ at $q$ is just $\emph{exp}^{nh}_q$.		
		\end{enumerate}	
\end{theorem}

After $ii)$ {and $iii)$} in Theorem \ref{main-theorem}, it is important to explicitly obtain Riemannian metrics on ${\mathcal U}_0 \subset \D_q$ such that Gauss condition holds for them.  A natural candidate for such a Riemannian metric is ${\mathcal G}_0 = ({\rm exp}^{nh}_q)^*(i_q^*g)$, where $i_q: {\mathcal M}^{nh}_q \to Q$ is the canonical inclusion. Unfortunately, ${\mathcal G}_0$ does not satisfy, in general, Gauss condition (see Example \ref{no-Gauss-condition}). However, the flat metric on ${\mathcal U}_0$ induced by a scalar product on $\D_q$ is a trivial example of a {Riemannian metric satisfying Gauss condition}. In fact, a general method to obtain examples of such metrics is presented in Remark \ref{Existence-metrics-GC} (see also Section \ref{examples} for the discussion in explicit examples of kinetic nonholonomic systems).   

\medskip

\noindent {\bf Organization of the paper.} The paper is structured as follows: in Section \ref{section-v}, we introduce some preliminary results necessary to prove Theorem \ref{main-theorem} as, for instance, the characterization of Riemannian metrics on vector spaces satisfying the so-called Gauss condition (see Definition \ref{gauss-condition}) given in Theorem  \ref{key1}. In Section \ref{section2} we review the construction of the nonholonomic exponential map constructed in \citep*{AMM} and we derive some of its main properties {for the particular case of kinetic nonholonomic systems.} By making use of this background, we prove the main result of this paper, Theorem \ref{main-theorem},  in Section \ref{section-proof}.  In Section \ref{examples}, we discuss several interesting examples of kinetic nonholonomic systems  $(Q, g, \D)$ and its associated Riemannian metrics whose radial geodesics are nonholonomic solutions. Finally, in Section \ref{conclusions}, we point out some of the important future lines of research opened up by the results of this article and, in particular, by Theorem \ref{main-theorem}.

\section{A result on Riemannian geometry of vector spaces}\label{section-v}

In this section, we will prove a result on Riemannian metrics in vector spaces which will be very useful in the proof of Theorem \ref{main-theorem}.

First of all, we will review a classical result in Riemannian geometry: the {\bf Gauss' Lemma} (see, for instance, \citep*{docarmo,O'Neill}).

Let $g$ be a Riemannian metric on a manifold $Q$ and $q$ a point in $Q$. Denote by ${\rm exp}^g_q: T_qQ \to Q$ the Riemannian exponential at the point $q$. As we know (see \citep*{docarmo,O'Neill}),
\begin{equation}\label{exponential-q0}
{\rm exp}^g_q(v_q) = c_{v_q}(1),
\end{equation}
for $v_q \in T_qQ$, where $c_{v_q}: [0, 1] \to Q$ is the unique geodesic in $Q$ with initial velocity $v_q$, that is, $c_{v_q}(0) = q$ and $\dot{c}_{v_q}(0) = v_q$. Note that ${\rm exp}^g_q(0_{q}) = q$. Moreover, there exist open subsets ${\mathcal U} \subseteq T_qQ$ and $U \subseteq Q$, with ${\mathcal U}$ starshaped about $0_q \in {\mathcal U}$ and $q \in U$, such that
\[
{\rm exp}^g_q: {\mathcal U} \to U
\]
is a diffeomorphism and 
\begin{equation}\label{prop-exp-map}
{\rm exp}^g_q(tv_q) = c_{v_q}(t), \; \; T_{0_q}{\rm exp}^g_q = id_{T_qQ}: T_{0_q}{\mathcal U} \simeq T_qQ \to T_qQ.
\end{equation}
In fact, the curve $t\in [0,1] \to c_{v_q}(t)\in Q$ is a minimizing geodesic from $q$ to $c_{v_q}(1)$. Then the
Gauss' lemma (see \citep*{docarmo,O'Neill}) implies that
\begin{equation}\label{G-lemma-1}
g({\rm exp}^g_q(v_q))((T_{v_q}{\rm exp}^g_q)(v_q)^{\bf v}_{v_q}, (T_{v_q}{\rm exp}^g_q)(w_q)^{\bf v}_{v_q}) = g(q)(v_q, w_q),
\end{equation}
for $v_q \in {\mathcal U}$ and $w_q \in T_qQ$, where $(v_q)^{\bf v}_{v_q}, (w_q)^{\bf v}_{v_q} \in T_{v_q}(T_qQ)$ are the vertical lifts to $TQ$ at $v_q$ of $v_q$ and $w_q$, respectively, {given by (\ref{vertical-lift-v})}. {So, under the linear identification $^{\bf v}_{v_q}: T_qQ \to T_{v_q}(T_qQ)$} between $T_qQ$ and $T_{v_q}(T_qQ)$, equation \eqref{G-lemma-1} gives
\begin{equation}
	g({\rm exp}^g_q(v_q))(T_{v_q}{\rm exp}^g_q (v_q),T_{v_q}{\rm exp}^g_q (w_q)) = g(q)(v_q, w_q).
\end{equation}

Thus, if we consider the Riemannian metric ${\mathcal G}_{0}$ on ${\mathcal U}$ defined by
\[
{\mathcal G}_{0} = ({\rm exp}^g_q)^*(g)
\]
then we deduce that
\[
{\mathcal G}_{0}(v_q)(v_q, w_q) = {\mathcal G}_{0}(0_q)(v_q, w_q).
\]
The previous fact motivates the definition below. Given a vector space  $E$ equipped with a Riemannian metric ${\mathcal G}$, if $u \in E$ then, as above, we will identify the tangent space $T_uE$ with $E$.    
\begin{definition}\label{gauss-condition}
We say that the  Riemannian manifold $(E, {\mathcal G})$ satisfies  the {\bf Gauss condition} if 
\[
{\mathcal G}(v) (v, w)={\mathcal G}(0)(v, w), \qquad \forall v, w\in E\; .
\]
The previous definition may also be applied to an open subset ${\mathcal U}$ of $E$ which contains the zero vector of $E$, that is, a Riemannian metric ${\mathcal G}$ on ${\mathcal U}$ satisfies {\bf the Gauss condition}  if 
\[
{\mathcal G}(v) (v, w)={\mathcal G}(0)(v, w), \qquad \forall v \in {\mathcal U} \mbox{ and } w\in E.
\]
\end{definition}
\begin{remark}\label{Existence-metrics-GC}{\rm If ${\mathcal G}$ is the flat metric on $E$ induced by a scalar product on $E$ then it is clear that ${\mathcal G}$ satisfies the Gauss condition. More generally, 
let $\overline{\mathcal G}$ be an arbitrary Riemannian metric on $E$ and ${\rm exp}^{\overline{\mathcal G}}_0: T_0E \to E$ the Riemannian exponential at $0$. As we know, there exist open subsets ${\mathcal U} \subseteq T_0E$ and $U \subseteq E$, such that $0 \in {\mathcal U} \cap U$ and
\[
{\rm exp}^{\overline{\mathcal G}}_0: {\mathcal U} \to U
\]
is a diffeomorphism. Then, proceeding as in the previous discussion to Definition \ref{gauss-condition}, we deduce that the Riemannian metric on ${\mathcal U}$ given by
\[
{\mathcal G} = \left(({\rm exp}^{\overline{\mathcal G}}_0)_{|{\mathcal U}}\right)^*(\overline{\mathcal G}_{|U})
\]
satisfies the Gauss condition. 
}
\end{remark}

{Now, let $\mathcal{G}$ be a Riemannian metric on a vector space $E$ and ${\rm exp}^{\mathcal G}_0: {\mathcal U} \subseteq T_0E \simeq E \to U\subseteq E$} the Riemannian exponential map at the zero vector $0 \in E$ (where we have used the canonical identification between $T_0E$ and $E$). Then, we denote by $r^{\mathcal G}_0: U \subseteq E \to \mathbb{R}$ the standard Riemannian radial function at $0$ for the Riemannian manifold $(E, {\mathcal G})$, that is, (see \citep*{docarmo,O'Neill}),
\[
r^{{\mathcal G}}_0(v) = \| ({\rm exp}^{\mathcal G}_0)^{-1}(v)\|_{{\mathcal G}(0)}, \; \; \mbox{ for } v \in U.
\] 
Moreover, using that the Riemannian manifold is a vector space, we can also define the {\bf radial distance function} 
	 $r^{{\mathcal G}}: E\rightarrow {\mathbb R}$ given by
	$$r^{{\mathcal G}}(v)=\|v\|_{{\mathcal G(v)}}=\sqrt{{\mathcal G}(v)(v, v)}  \;.$$
	In general, we have that $r^{\mathcal G}_0 \neq (r^{\mathcal G})_{|U}$. However, if $(E, {\mathcal G})$ satisfies the Gauss condition in ${\mathcal U}$, we will see that that ${\rm exp}^{\mathcal G}_0: {\mathcal U} \subseteq E \to E$ is the canonical inclusion of ${\mathcal U}$ in $E$ (see Theorem \ref{key1} below) and, thus, $r^{\mathcal G}_0 = (r^{\mathcal G})_{|U}$. 
	
First, we will prove the following result.
\begin{lemma}\label{lema1}
	The radial distance function $r^{{\mathcal G}}: E \to \mathbb{R}$ is smooth on $E\setminus\{0\}$. Moreover, if $(E, {\mathcal G})$ satisfies  the Gauss condition, then the gradient vector field $\emph{grad}_{{\mathcal G}} r^{{\mathcal G}}\Big|_{E\setminus\{0\}}$ of $r^{{\mathcal G}}$ on $E\setminus\{0\}$ is given by
	\[
	\emph{grad}_{{\mathcal G}} r^{{\mathcal G}}(v)=\frac{v}{\,\ \|v\|_{\mathcal G(0)}}, \; \; \mbox{ for } v \in E\setminus\{0\}
    \]
and, in addition, it is	unitary relative to the metric ${\mathcal G}$, that is, 
\begin{equation*}
	\|\emph{grad}_{{\mathcal G}} r^{{\mathcal G}}(v)\|_{{\mathcal G(v)}}=1, \; \;   \mbox{ for } v \in E\setminus\{0\}.
\end{equation*}
	\end{lemma}

\begin{proof}
	The first part of the lemma is obvious. On the other hand, using the definition of the gradient vector field and the Gauss condition, we have that, for $v \in E\setminus\{0\}$ and $u \in E$,
	\begin{eqnarray*}
{\mathcal G}(v)(\text{grad}_{{\mathcal G}} r^{{\mathcal G}}(v), u) &=& 
	\langle dr^{{\mathcal G}}(v), u\rangle 
	=\langle dr^{{\mathcal G}}(v), \frac{d}{dt}\Big|_{t=0} (v+tu)\rangle \\
	&=&\frac{d}{dt}\Big|_{t=0} r^{\mathcal G}(v+tu)
	=\frac{d}{dt}\Big|_{t=0} \sqrt{\left({\mathcal G}(0)(v+tu,v+tu\right)}\\
	&=&\frac{d}{dt}\Big|_{t=0} \sqrt{
		\|v\|^2_{\mathcal G(0)}+2t {\mathcal G}(0)(u, v)+t^2	\|u\|^2_{\mathcal G(0)}}\\
	&=&\frac{{\mathcal G}(0)(v, u)}{\|v\|_{\mathcal G(0)}}\\
	&=&{\mathcal G}(v)\left( \frac{v}{\|v\|_{\mathcal G(0)}} , u\right)
	\end{eqnarray*}
	This implies that, 
	\[
	\text{grad}_{{\mathcal G}} r^{{\mathcal G}}(v)=\frac{v}{\|v\|_{\mathcal G(0)}}.
	\]
	Now, using again the Gauss condition, it follows that
	\[
	{\mathcal G}(v)(\frac{v}{\|v\|_{\mathcal G(0)}}, \frac{v}{\|v\|_{\mathcal G(0)}}) 
	= \frac{1}{({\|v\|_{{\mathcal G}(0)}})^2}{\mathcal G}(0)(v, v) = 1
	\]
	which concludes the proof of the result. 
\end{proof}

	Another fact which is relevant for our purposes is the following.

\begin{lemma}\label{lema2}
{If $(E, {\mathcal G})$ satisfies  the Gauss condition then   
the} integral curves of the vector field $\; U = \emph{grad}_{{\mathcal G}} r^{{\mathcal G}}$ are geodesic.
\end{lemma}

\begin{proof}
This follows using that $U$ is unitary and the gradient vector field of a real $C^{\infty}$-function on $E\setminus\{0\}$. 

In fact, if $\nabla$ is the Levi-Civita connection  of ${\mathcal G}$ and $X \in \frak{X}(E\setminus\{0\})$ then, since $\nabla$ is metric and torsion free, we have that
\[
{\mathcal G}(\nabla_UU, X) = U({\mathcal G}(U, X))- {\mathcal G}(U, \nabla_UX) = U(X(r^{{\mathcal G}})) - {\mathcal G}(U, [U, X]) + {\mathcal G}(U, \nabla_XU).
\]
On the other hand, using again that $\nabla$ is metric and that $U$ is unitary, we deduce that
\[
0 = X({\mathcal G}(U, U)) = 2{\mathcal G}(U, \nabla_X U).
\]
Thus, we obtain that
\begin{equation}\label{U-geodesics}
\begin{array}{rcl}
{\mathcal G}(\nabla_UU, X) &=&  U(X(r^{{\mathcal G}})) - [U, X](r^{{\mathcal G}}) \\ &= & U(X(r^{{\mathcal G}})) - U(X(r^{{\mathcal G}})) + X(U(r^{{\mathcal G}})) = 0, 
\end{array}
\end{equation}
where the last equality follows using that
\[
U(r^{{\mathcal G}}) = dr^{{\mathcal G}}(U) = {\mathcal G}(U, U) = 1.
\]
Finally, relation (\ref{U-geodesics}) implies that
$\nabla_U U = 0$
and, therefore, the integral curves of $U$ are geodesic.
\end{proof}

Now, we are able to prove the result we announced at the beginning of the section.	
\begin{theorem}\label{key1}
If ${\mathcal G}$ is a Riemannian metric on a vector space $E$ and
\[
{\rm exp}^{\mathcal G}_0: {\mathcal U} \subseteq T_0E \simeq E \to E,
\]
is the exponential map at the zero vector then
the following conditions are equivalent: 
\begin{itemize}
	\item[i)]  The map ${\rm exp}^{\mathcal G}_0: {\mathcal U}\subseteq E \to E$ is the canonical inclusion of ${\mathcal U}$ in $E$.
	\item[ii)] For each $u \in {\mathcal U}$,  the line
	\[
	t \in [0, 1] \to tu \in {\mathcal U}
	\]
starting at  the zero vector is a minimizing geodesic for $(E, {\mathcal G})$ with initial velocity $u$.
	\item[iii)] $(E, {\mathcal G})$ satisfies the Gauss condition in ${\mathcal U}$.
\end{itemize}
\end{theorem}
	\begin{proof}$\,$
As we know,
\[
{\rm exp}^{\mathcal G}_0(tu) = c_u(t), \; \; \mbox{ for } t\in [0, 1]
\]
where $c_u: [0, 1] \to E$ is the minimizing geodesic with initial velocity $u \in {\mathcal U}$. 		
$[i)\Leftrightarrow ii)]$ If ${\rm exp}^{\mathcal G}_0: {\mathcal U} \to E$ is the inclusion of ${\mathcal U}$ in E then $c_u(t) = tu$ which proves $i) \Rightarrow ii)$.

Conversely, if $ii)$ holds then it is clear that
\[
{\rm exp}^{\mathcal G}_0(tu) = tu, \; \; \mbox{ for } u\in {\mathcal U},
\]
and, thus, ${\rm exp}^{\mathcal G}_0: {\mathcal U}\subseteq E \to E$ is the canonical inclusion of ${\mathcal U}$ in $E$.

$[i)\Rightarrow iii)]$
If $u \in {\mathcal U}$ and $v \in E$ then, using the Gauss' Lemma and $i)$,  we have that 
\begin{eqnarray*}
	{\mathcal G}(u)(u,v)&=&	{\mathcal G}({\rm exp}^{\mathcal G}_0(u))
	(T_u{\rm exp}^{\mathcal G}_0 (u), T_u{\rm exp}^{\mathcal G}_0 (v))\\
	&=&{\mathcal G}(0)(u,v).
\end{eqnarray*}
So, $(E, {\mathcal G})$ satisfies the Gauss condition in ${\mathcal U}$. 

$[iii)\Rightarrow i)]$ Let $r^{{\mathcal G}}: E \to \mathbb{R}$ be the radial distance function and $U = \text{grad}_{{\mathcal G}} r^{{\mathcal G}}$. From Lemma \ref{lema1}, we have that $U(u) = \frac{u}{\|u\|_{\mathcal G(0)}}$, for  $ u\in E\setminus\{0\}$. This implies that the line $l_{\frac{u}{\|u\|_{{\mathcal G}(0)}}}: (0, \infty) \to E \setminus\{0\}$ given by
\[
l_{\frac{u}{\|u\|_{{\mathcal G}(0)}}}(t) = t \frac{u}{\|u\|_{\mathcal{G}(0)}}
\]
is an integral curve of $U$. Then, using Lemma \ref{lema2}, $l_{\frac{u}{\|u\|_{{\mathcal G}(0)}}}$ is a geodesic. Thus, the homothetic reparametrization of $l_{\frac{u}{\|u\|_{{\mathcal G}(0)}}}$ defined by
\[
t \to  l_{\frac{u}{\|u\|_{{\mathcal G}(0)}}}(\|u\|_{{\mathcal G}(0)} t)= tu
\]
is also a geodesic. It is clear that the curve defined above is continuously extendible to $t = 0$. So, from Lemma 8 of Chapter 5 in \citep*{O'Neill}, it is also extendible as a geodesic. Moreover, its initial velocity is $u$. Thus
\[
{\rm exp}^{{\mathcal G}}_0(tu) = tu, \; \; \mbox{ for } t\in [0, 1].
\]
This proves $i)$.
\end{proof}

\section{Nonholonomic exponential map for kinetic nonholonomic systems}\label{section2}
In this section, we will review the definition of the nonholonomic exponential map for a kinetic nonholonomic system and we will obtain some results on this map (for the definition of the nonholonomic exponential map associated with an arbitrary nonholonomic system, see \citep*{AMM}).

First of all, we will see that the solutions of the equations of motion of a kinetic nonholonomic system are the geodesics of a  constrained connection (the nonholonomic connection) on the configuration space restricted to initial conditions in $\D$. This construction seems to have been first made in \citep*{Synge28}.

As we have commented in the introduction, a kinetic nonholonomic system is determined by the triple $(Q, g, \D)$, where $Q$ is a finite dimensional smooth manifold, $g$ is a Riemannian metric on $Q$ and $\D$ is a nonintegrable distribution determining the nonholonomic constraints. 

The nonholonomic connection $\nabla^{nh}$ is   defined as
\begin{equation}\label{nhconnection}
\nabla^{nh}_{X} Y:=P(\nabla_{X}^{g} Y)+\nabla^{g}_{X}[P'(Y)], \; \; \mbox{ for } X, Y \in \frak{X}(Q),
\end{equation}
where  $P:TQ\rightarrow \D$ is the associated  orthogonal projector onto the distribution $\D$ and $P':TQ\rightarrow\D^{\bot}$ is the orthogonal projector onto $\D^{\bot}$, the orthogonal distribution.

This connection is not symmetric in general neither it is compatible with the metric. Nevertheless, it satisfies the more restricted condition of compatibility with the Riemannian metric $g$ over sections of $\D$ (see \citep*{Lewis98}), i.e.,
\begin{equation}\label{Dcompatibility}
X(g(Y,Z))=g(\nabla^{nh}_{X} Y,Z)+g(Y,\nabla^{nh}_{X} Z), \quad \forall X,Y,Z\in\Gamma(\D).
\end{equation}
It is interesting to note that if   $Y\in\Gamma(\D)$ then $\nabla^{nh}_{X} Y=P(\nabla_{X}^{g} Y)\in \Gamma(\D)$ for any vector field $X\in\mathfrak{X}(Q)$.

The  geodesics $c$ for this connection which satisfy the constraints, that is, 
\begin{equation}\label{cnh-1}
\nabla^{nh}_{\dot{c}(t)}\dot{c}(t)=0\; ,\qquad   \dot{c}(0)\in \D_{c(0)}
\end{equation}
are precisely the solutions of the nonholonomic problem given by $(Q, g, \D)$ (see, for instance, \citep*{Lewis98, BLMMM2011}). Observe that this equation is equivalent to Equations (\ref{cnh}).

\begin{lemma}\label{homogeneity}
{Let $c_v: I \to Q$ be a nonholonomic geodesic with initial velocity $v \in \D_q$, i.e.}
\[
c_{v}(t_0) = q \; \; \mbox{ and } \; \; \dot{c}_{v}(t_0) = v.
\]
\begin{enumerate}
\item
{We have that
\begin{equation}\label{constant-norm}
\|\dot{c}_v(t) \|_{g(c_v(t))} = \| v\|_{g(q)}, \; \; \mbox{ for } t \in I.
\end{equation}}
\item
If $v = 0$ then $c_v(t) = q$, for every $t \in I$.
\item
If $v \neq 0$ then a reparametrization of $c_v$,
\[
c_v \circ h: J \to Q, \; \; s \to c_v(h(s))
\]
is a nonholonomic geodesic if and only if
\[
h(s) = as + b, \; \; \mbox{ with } a, b \in \mathbb{R}.
\]
\end{enumerate}
\end{lemma}
\begin{proof}
On one hand, using (\ref{Dcompatibility}) and the fact that $\dot{c}_v(t) \in \D_{c(t)}$, for every $t \in I$, it follows that
\[
\frac{d}{dt}\left(g(c_v(t))(\dot{c}_v(t), \dot{c}_v(t))\right) = 2g(c_v(t))(\nabla^{nh}_{\dot{c}_v(t)}\dot{c}_v(t), \dot{c}_v(t)) = 0.
\]
Thus, we deduce that
{
\[
\|\dot{c}_v(t) \|_{g(c_v(t))} = \|\dot{c}_v(0)\|_{g(q)} = \| v\|_{g(q)}, \; \; \mbox{ for } t \in I.
\]}
This proves item $1$. Item $2$ follows from $1$.

On the other hand, if $v \neq 0$ then considering the reparametrization $c_v \circ h: J \to Q$, we have that
\begin{eqnarray*}
\nabla^{nh}_{\frac{d}{ds}(c_v \circ h)}\frac{d}{ds}\left(c_v \circ h\right) 
& = & \nabla^{nh}_{\frac{dh}{ds}(\frac{dc_v}{dt} \circ h)} \frac{dh}{ds}\left(\frac{dc_v}{dt} \circ h\right)\\
 &= & 
\left(\frac{dh}{ds}\right)^2 \nabla^{nh}_{(\frac{dc_v}{dt} \circ h)} \left(\frac{dc_v}{dt} \circ h \right) + \frac{d^2h}{ds^2} \left(\frac{dc_v}{dt} \circ h \right).
\end{eqnarray*}

Now, from (\ref{constant-norm}), we obtain that
\[
\left(\frac{dc_v}{dt} \circ h \right)(s) \neq 0, \; \; \forall s.
\]
Therefore, since
$\nabla^{nh}_{(\frac{dc_v}{dt} \circ h)} (\frac{dc_v}{dt} \circ h) = 0$,
we conclude that
\[
\nabla^{nh}_{\frac{d}{ds}(c_v \circ h)}\frac{d}{ds}(c_v \circ h) = 0 \Leftrightarrow \frac{d^2h}{ds^2} = 0 \Leftrightarrow h(s) = as + b,
\]
with $a, b \in \mathbb{R}$.
\end{proof}
The tangent lifts of the nonholonomic geodesics of a kinetic nonholonomic system $(Q, g, \D)$ are the integral curves of a vector field of  $\Gamma_{(g,\D)}\in {\mathfrak X}(\D)$, which is a second-order differential equation along the points of $\D$, considered as a vector subbundle of $TQ$ (see, for instance, \citep*{LMdD1996}). 

Denote by $\phi_t^{\Gamma_{(g,\D)}}: \D \rightarrow \D$ the flow of $\Gamma_{(g,\D)}$ and for a sufficiently small positive number $h$, we consider the open subset of $\D$ given by
\begin{equation*}
M_{h}^{\Gamma_{(g,\D)}}=\{ v\in\D \ | \ \phi_{t}^{\Gamma_{(g,\D)}}(v) \ \text{is defined for} \ t\in [0,h] \}.
\end{equation*}

Using the last part of Lemma \ref{homogeneity} we can assume, without the loss of generality, that $h=1$. Then, we will denote the open subset 
$M_{1}^{\Gamma_{(g,\D)}}$ of $\D$ by $M^{\Gamma_{(g,\D)}}$. {In addition, from the second part of Lemma \ref{homogeneity}, we also have that the zero section in $\D$ is contained in $M^{\Gamma_{(g,\D)}}$.}

From the flow of $\Gamma_{(g,\D)}$, we can define the nonholonomic exponential map
	\begin{align*}
	\text{exp}^{\Gamma_{(g,\D)}}:M^{\Gamma_{(g,\D)}}\subseteq \D & \rightarrow  Q \times Q\\
	v & \mapsto(\tau_Q(v), \tau_{Q}\circ\phi_1^{\Gamma_{(g,\D)}}(v))
	\end{align*}
(see \citep*{AMM}).	
	We remark that if $c_{v}: [0, 1]\rightarrow Q$ is the nonholonomic geodesic with $\dot{c}_{v}(0)=v$ then
		\begin{equation}\label{nh-exponential-map}
		\text{exp}^{\Gamma_{(g,\D)}}(v)=(\tau_Q(v), c_{v}(1)).
	         \end{equation}
	
	We will use in the sequel the restriction of this map to the open subset $M_{q}^{\Gamma_{(g,\D)}} = M^{\Gamma_{(g,\D)}} \cap \D_{q} $ of $\D_{q}$ with $q\in Q$ fixed, that is, we define
	\[
	\text{exp}_{q}^{nh}=	\text{pr}_2 \circ \text{exp}^{\Gamma_{(g,\D)}}\Big|_{M_{q}^{\Gamma_{(g,\D)}}}: M_{q}^{\Gamma_{(g,\D)}} \subset \D_{q}\longrightarrow Q
	\]
So, if $v_q \in \D_{q}$ and $c_{v_q}: [0, 1] \to Q$ is the nonholonomic geodesic with initial velocity $v_q$ then
\[
\text{exp}_{q}^{nh}(v_q) = c_{v_q}(1).
\]
The reader is invited to compare the definition of $\text{exp}_{q}^{nh}$ with that of the Riemannian exponential at $q$ (see (\ref{exponential-q0})).

In fact, the nonholonomic exponential map conserves many of the properties we may find in Riemannian exponential maps. {For instance, verify that }
\begin{equation}\label{first-property}
\text{exp}_{q}^{nh}(0_{q}) = q.
\end{equation}
The second result we are going to prove is the \textit{rescaling lemma}.

\begin{lemma}\label{rescaling:lemma}
	Let $c_{v_q}:[0,1]\rightarrow Q$ denote the nonholonomic geodesic with initial velocity $v_q\in M_{q}^{\Gamma_{(g,\D)}}$. {Then, for $t \in [0,1]$, we have that $tv_q \in M_{q}^{\Gamma_{(g,\D)}}$ and}
	\begin{equation*}
	c_{v_q}(t)=\emph{exp}^{nh}_{q}(tv_q).
	\end{equation*}
\end{lemma}

\begin{proof}
If $t = 0$ the result is obvious.

Suppose that $t \neq 0$. Then, we can consider the curve
\begin{equation*}
	{s \in \left[0, \frac{1}{t}\right]  \mapsto q(s)=c_{v_q}(ts) \in Q.}
\end{equation*}
Using the last part of Lemma \ref{homogeneity}, we deduce that the previous curve is a nonholonomic geodesic. Moreover, its initial velocity is $tv_q$. Thus,
\[
c_{tv_q}(s) = c_{v_q}(ts).
\]
As a consequence,  $tv_q \in M_{q}^{\Gamma_{(g,\D)}}$ and, in addition,	 by the definition of exponential map, it follows that
	\begin{equation*}
	{\rm exp}^{nh}_{q}(tv_q)=c_{tv_q}(1)=c_{v_q}(t).
	\end{equation*}
\end{proof}

We also have that (see \citep*{AMM} for a more general situation) 
\begin{proposition}\label{Propo1} If $q\in Q$ then, under the canonical linear identification between $T_{0_{q}}M_{q}^{\Gamma_{(g,\D)}}$ and $\D_{q}$, the linear map
\[
T_{0_q}\emph{exp}^{nh}_q: T_{0_{q}}M_{q}^{\Gamma_{(g,\D)}} \simeq \D_q \to T_qQ 
\]
is just the canonical inclusion of $\D_q$ in $T_qQ$. So,
there exists  an open subset $\mathcal{U}_{0}$ of $\D_{q}$ containing $0_{q}$, which we can choose to be starshaped about $0_{q}$, such that the nonholonomic exponential map ${\rm exp}_{q}^{nh}:\mathcal{U}_{0}\rightarrow  Q$ is an  embedding.   
\end{proposition}
\begin{proof}
Observe that for $v_q \in D_q$,
\begin{eqnarray*}
(T_{0_q}{\rm exp}^{nh}_{q})(v_{q})&=& \frac{d}{dt}\Big|_{t=0} {\rm exp}^{nh}_{q}(tv_{q})\\
&=& \frac{d}{dt}\Big|_{t=0}c_{v_q}(t)=v_q
\end{eqnarray*}
using Lemma \ref{rescaling:lemma}.
Therefore, $T_{0_q} {\rm exp}^{nh}_{q}: \D_{q}\rightarrow  T_{q}Q$ is just the canonical inclusion.

Thus, there exists an open subset $\mathcal{U}_{0}$ of $D_{q_0}$, with $0_{q_0} \in \mathcal{U}_{0}$, such that ${\rm exp}^{nh}_{q}: \mathcal{U}_{0} \to Q$ is a diffeomorphism over its image. {In fact, we can choose $\mathcal{U}_0$ to be starshaped about $0_{q}$.}
\end{proof}

\section{Proof of Theorem \ref{main-theorem}}\label{section-proof}

We will prove each one of the items in Theorem \ref{main-theorem}

\begin{proof}$\,$
$i)$ Take in Proposition \ref{Propo1}
\[
{\mathcal M}^{nh}_q = \text{exp}^{nh}_q({\mathcal U}_0).
\] 	
Then, using (\ref{first-property}), Lemma \ref{rescaling:lemma} and Proposition \ref{Propo1}, we deduce the first part of the theorem.	

$ii)$ Let $t \mapsto c_{u_q}(t)$ be a radial kinetic nonholonomic trajectory departing from $q$, that is, $c_{u_q}$ is a nonholonomic trajectory and
\[
c_{u_q}(0) = q, \; \; \dot{c}_{u_q}(0) = u_q \in \D_q.
\]
Then, using that ${\mathcal U}_0$ is an open subset of $\D_q$ and $0_q \in {\mathcal U}_0$, {there exists $v_{q}\in {\mathcal U}_0$ and a real number $\lambda > 0$ such that $v_q = \displaystyle \frac{u_q}{\lambda}$. Also, by item $i)$ of this Theorem}, the radial kinetic nonholonomic trajectory
\[
c_{v_q}: [0, 1] \to Q, \; \; c_{v_q}(t) = \text{exp}^{nh}_q(tv_q)
\]
is contained in ${\mathcal M}^{nh}_q$. As 
\[
\dot{c}_{v_q}(0) = v_q = \frac{u_q}{\lambda} = \frac{\dot{c}_{u_q}(0)}{\lambda},
\] 
from Lemma \ref{homogeneity}, $c_{u_q}$ is just the homothetic reparametrization of $c_{v_q}$ given by
\[
c_{u_q}(t) = c_{v_q}(\lambda t), \; \; \mbox{ for } 0 < t < \frac{1}{\lambda}.
\]
This proves the first part of $ii)$ in the  theorem.

Now, suppose that $g^{nh}_q$ is a Riemannian metric on ${\mathcal M}^{nh}_q$ and that ${\mathcal G}_0$ is the Riemannian metric on ${\mathcal U}_0$ given by ${\mathcal G}_0 = (\text{exp}^{nh}_q)^*(g^{nh}_q)$. From Theorem \ref{key1}, it follows that the lines through $0_q$, which are of the form
\[
t \in [0, 1] \mapsto tv_q, \; \; \mbox{ for } v_q \in {\mathcal U}_0,
\]
are minimizing geodesics in the Riemannian manifold $({\mathcal U}_0, {\mathcal G}_0)$ if and only if ${\mathcal G}_0$ satisfies the Gauss condition. 

On the other hand, from the definition of ${\mathcal G}_0$,  we have that $\text{exp}^{nh}_q: ({\mathcal U}_0, {\mathcal G}_0) \to ({\mathcal M}^{nh}_q, g^{nh}_q)$ is an isometry and, from Lemma \ref{rescaling:lemma}, the image by $\text{exp}^{nh}_q$ of the lines through $0_q$ are just the radial kinetic nonholonomic trajectories departing from $q$. Thus, we conclude that these trajectories are minimizing geodesics in $({\mathcal M}^{nh}_q, g^{nh}_q)$ if and only if ${\mathcal G}_0$ satisfies the Gauss condition in ${\mathcal U}_0$.

$iii)$ As we know, there exist Riemannian metrics on ${\mathcal U}_0$ which satisfy the Gauss condition (see Remark \ref{Existence-metrics-GC}). So, if ${\mathcal G}_0$ is one of them and we define the Riemannian metric $g^{nh}_q$ on ${\mathcal M}^{nh}_q$ to be given by
\[
g^{nh}_q = \left((\text{exp}_q^{nh})^{-1}\right)^*({\mathcal G}_0),
\]
it is clear that, using item $ii)$ in this theorem, the radial kinetic nonholonomic trajectories departing from $q$ are minimizing geodesics in the Riemannian manifold $({\mathcal M}^{nh}_q, g^{nh}_q)$. This proves the first part of item $iii)$.

Now, under the canonical linear identification between $T_q{\mathcal M}^{nh}_q$ and $\D_q$ induced by the linear isomorphism
\[
T_{0_q}\text{exp}^{nh}_q: \D_q \to T_q{\mathcal M}^{nh}_q,
\]
let ${\rm exp}^{g^{nh}_q}_q: T_q{\mathcal M}^{nh}_q \simeq \D_q \to {\mathcal M}^{nh}_q$ be the Riemannian exponential associated  with $g^{nh}_q$ at the point $q$. We may assume, without loss of generality, that 
\[
{\rm exp}^{g^{nh}_q}_q: {\mathcal U}_0 \subseteq \D_q \to {\mathcal M}^{nh}_q
\]
is a diffeomorphism. Moreover, if $v_q \in {\mathcal U}_0$ then, since the radial kinetic nonholonomic trajectory $c_{v_q}: [0, 1] \to {\mathcal M}^{nh}_q$ is a minimizing geodesic in the Riemannian manifold $({\mathcal M}^{nh}_q, g^{nh}_q)$ with initial velocity $v_q$, we deduce that
\[
{\rm exp}^{g^{nh}_q}_q(tv_q) = c_{v_q}(t) = \text{exp}^{nh}_q(tv_q), \; \; \mbox{ for } t \in [0, 1].
\]
So, ${\rm exp}^{g^{nh}_q}_q = \text{exp}^{nh}_q$.
	 \end{proof}

\section{Examples}\label{examples}	
First of all, we will see that if $i_q: {\mathcal M}^{nh}_q \to Q$ is the canonical inclusion then the Riemannian metric
\[
{\mathcal G}_0 = (\text{exp}^{nh}_q)^*(i_q^*g)
\]
on ${\mathcal U}_0 \subseteq \D_q$ doesn't satisfy, in general, Gauss condition. 	 
		\begin{example}\label{no-Gauss-condition}
	Consider the nonholonomic particle in $Q=\R^{3}$, that is, $g$ is the standard flat Riemannian metric on $\R^3$ and $\D$ is the constraint distribution determined by
	\begin{equation*}
	\D=\{v\in TQ \ | \ \dot{z}=y\dot{x}\}.
	\end{equation*}
	Here, $(x, y, z)$ and $(x, y, z,\dot{x},\dot{y},\dot{z})$ are the standard coordinates on $\R^3$ and $T\R^{3}$, respectively.
	
	For simplicity we will fix $q= 0 = (0, 0, 0)$. It is clear that
	\[
	\D_0 = < u_0 = \frac{\partial}{\partial x}_{|0}, v_{0}=\frac{\partial}{\partial y}_{|0} >.
	\]
	Denote by $(u, v)$ the linear coordinates on $\D_0$ induced by the previous basis.
	
	 The nonholonomic exponential map $\text{exp}^{nh}_{0}:\D_{0}\rightarrow Q$ is known to be given by {(see \cite{AMM})}:
	\begin{equation*}
		\text{exp}^{nh}_{0}(u,v)=\left(\frac{u}{v}\arcsinh(v),v,\frac{u}{v}(\sqrt{v^{2}+1}-1)\right)
	\end{equation*}
	if $v\neq 0$ and
	\begin{equation*}
		\text{exp}^{nh}_{0}(u,0)=\left(u,0,0\right), \quad \text{if} \ v=0.
	\end{equation*}
	The tangent map of $\text{exp}^{nh}_{0}$ at $u_{0}$ is represented in coordinates by the matrix
	\begin{equation*}
		T_{u_{0}}\text{exp}^{nh}_{0}= \left(\begin{matrix}
		1 & 0 \\
		0 & 0 \\
		1 & \frac{1}{2}
		\end{matrix}\right).
	\end{equation*}
	Thus, it follows that
	\begin{equation*}
	\begin{split}
		g(\text{exp}^{nh}_{0}(u_{0})) & \left(T_{u_{0}}\text{exp}^{nh}_{0}(u_{0}),T_{u_{0}}\text{exp}^{nh}_{0}((v_{0}))\right) \\
		& =g(\text{exp}^{nh}_{0}(u_{0}))\left( \frac{\partial}{\partial x}+\frac{\partial}{\partial z},\frac{1}{2}\frac{\partial}{\partial z}\right) =\frac{1}{2}.
	\end{split}
	\end{equation*}
	However, we have that
	\begin{equation*}
		g(0)\left(u_{0},v_{0}\right)=0.
	\end{equation*}
	Thus, the Riemannian metric ${\mathcal G}_0 = (\text{exp}^{nh}_0)^*(i_0^*g)$ on $\D_0$ does not satisfy Gauss condition.
\end{example}

Next, for a fixed point $q\in Q$, we will give examples of { Riemannian metrics ${\mathcal G}_0$ satisfying Gauss condition} on $\D_q$ and we will obtain the corresponding metrics
\[
g^{nh}_q = \left((\text{exp}^{nh}_q)^{-1}\right)^*({\mathcal G}_0)
\]
on ${\mathcal M}^{nh}_q$.  

\begin{example}
	Consider again the nonholonomic particle in $Q=\R^{3}$ and fix the point $q = 0 \in \R^3$. 
	
	Let ${\mathcal G}_0$ be the standard flat metric in $\D_0 \simeq \R^{2}$  so that
	\begin{equation*}
	\left(	({\mathcal G}_0)_{ij}(u, v)\right)=\left(\begin{matrix}
		1 & 0 \\
		0 & 1
		\end{matrix}\right).
	\end{equation*}
	It is clear that ${\mathcal G}_0$ satisfies the Gauss condition in $\D_0$. 	
Denote by $(x, y)$ the coordinates on ${\mathcal M}^{nh}_0$ induced by the coordinates $(u, v)$ on $\D_0$ and by the nonholonomic exponential map $\text{exp}^{nh}_0$. Since
	\begin{equation*}
		{\mathcal G}_0 = du\otimes d u + d v \otimes d v,
	\end{equation*}
	and 
	\begin{equation*}
		(\text{exp}^{nh}_0)^{-1}(x,y)= \left( \frac{xy}{\arcsinh(y)},y \right)
	\end{equation*}
	we have that
	\begin{equation*}
		g_{0}^{nh}=E dx \otimes dx + F dx\otimes dy + F dy \otimes dx + G dy \otimes d y,
	\end{equation*}
	with
	\begin{equation*}
		\begin{split}
			& E=\frac{y^2}{\arcsinh^{2}(y)}, \quad F=\frac{xy (\arcsinh(y) \sqrt{y^2+1}-y)}{\sqrt{y^2+1} \arcsinh^{3}(y)} \\
			& G = \frac{-2 \arcsinh(y) \sqrt{y^2+1} x^2 y+ \arcsinh^2(y) (y^2+1) x^2+x^2 y^2}{\arcsinh^4(y) (y^2+1)}+1.
		\end{split}
	\end{equation*}
\end{example}

\begin{example}
	For the same nonholonomic system, one could choose other Riemannian metric satisfying the Gauss condition. Consider the Riemannian metric on $\D_{0}$ given by
	\begin{equation*}
		\mathcal G_{0} =(1-v^{2})du\otimes d u+ uv du\otimes dv + uv dv\otimes du + (1-u^{2})d v \otimes d v.
	\end{equation*}
	Note that, this tensor is degenerate on the unitary circle around the origin of $\D_{0}$ (where the radius is measured with respect to the euclidean metric). To overcome this technicality, we will restrict ourselves to the open ball with unit radius on which the metric $\mathcal G_{0}$ is non-degenerate. This example illustrates that Theorem \ref{main-theorem} could in principle be extended to convex subsets of vector spaces.
	
	The Chrystoffel symbols with respect to this metric are
	\begin{equation*}
		\begin{split}
			& \Gamma^u_{uu} = \frac{2 u v^2}{u^2 + v^2 - 1}, \quad \Gamma^u_{uv} = -\frac{(2u^2 - 1)v}{u^2 + v^2 - 1}, \quad \Gamma^u_{vv} = \frac{2(u^3 - u)}{u^2 + v^2 - 1}, \\
			& \Gamma^v_{uu} = \frac{2 (v^3 - v)}{u^2 + v^2 - 1}, \quad \Gamma^v_{uv} = -\frac{(2uv^2 - u)}{u^2 + v^2 - 1}, \quad	\Gamma^v_{vv} = \frac{2u^2 v}{u^2 + v^2 - 1}.
		\end{split}
	\end{equation*}
	Consider now the lines $c_{(u_{0},v_{0})}:[0,1]\rightarrow \D_{0}$ contained in the unitary open ball in $\D_{0}$ departing from $0$, with the coordinate expression
	\begin{equation*}
		c_{(u_{0},v_{0})}(t)=(u_{0}t,v_{0}t).
	\end{equation*}
	It is easy to check that the curves $c_{(u_{0},v_{0})}$ satisfy the geodesic equations and therefore the lines through the origin are in fact geodesics. At the same time, it is clear that the exponential map is the identity and we can check that the Riemannian metric $\mathcal G_{0}$ satisfies the Gauss condition: let $X,Y\in\D_{0}$ with the local expression $(X^{u},X^{v})$ and $(Y^{u},Y^{v})$, respectively.
	
	Then
	\begin{equation*}
		\begin{split}
			\mathcal G_{0}(X)(X,Y) = & (1-(X^{v})^{2})X^{u}Y^{u}+X^{u}X^{v}(X^{u}Y^{v}+X^{v}Y^{u}) \\
			& +(1-(X^{u})^{2})X^{v}Y^{v} \\
			= & X^{u}Y^{u}+ X^{v}Y^{v} = \mathcal G_{0}(0)(X,Y).
		\end{split}
	\end{equation*}
\end{example}

\begin{example}\label{vertical-rolling-disk}
	Consider the example of the vertical rolling disk with $Q=\R^{2}\times \Es^{1} \times \Es^{1}$, parametrized by the coordinates $(x,y,\theta,\varphi)$. This system is described by the Lagrangian function $L:TQ\rightarrow \R$ given by
	\begin{equation*}
	L(x,y,\varphi,\theta,\dot{x},\dot{y},\dot{\varphi},\dot{\theta})=\frac{1}{2}m(\dot{x}^2+\dot{y}^2)+\frac{1}{2}I\dot{\theta}^2+\frac{1}{2}J\dot{\varphi}^2,
	\end{equation*}
	and subjected to the constraint distribution $\D\subseteq TQ$ determined by the equations
	\begin{equation*}
	\dot{x}=R \cos \varphi \ \dot{\theta}, \quad \dot{y}= R \sin \varphi \ \dot{\theta},
	\end{equation*}
	where $R$ is the radius of the disk, $m$ is the mass of the disk and $I$, $J$ are moments of inertia about an axis perpendicular to the plane of the disk and contained in the plane of the disk, respectively.
	
	In \citep*{BLOCH2009225}, the authors considered the Lagrangian function $L^{mod}:TQ\rightarrow \R$ given by
	\begin{equation*}
		L^{mod}=-\frac{1}{2}m(\dot{x}^2+\dot{y}^2)+\frac{1}{2}I\dot{\theta}^2+\frac{1}{2}J\dot{\varphi}^2 + mR\dot{\theta}(\cos \varphi \dot{x}+\sin \varphi \dot{y})
	\end{equation*}
	and showed that the trajectories of the Euler-Lagrange equations for $L^{mod}$ with initial velocity in the distribution $\D$ are exactly the nonholonomic trajectories for $(L,\D)$.
	
	For simplicity we will assume that both the mass and the radius are unitary $m=R=1$ and we will consider trajectories departing from the point $q$ with coordinates $(0,0,0,0)$. We will show that the metric $g^{mod}$ associated to the kinetic Lagrangian $L^{mod}$ is related to a metric on $\D_{q}$ satisfying the Gauss {condition}.
	
	Indeed, the Lagrangian $L^{mod}$ is of the type
	\begin{equation*}
		L^{mod}(v)=\frac{1}{2}g^{mod}(v,v),
	\end{equation*}
	where
	\begin{equation*}
		(g^{mod})_{ij}=\left(
		\begin{matrix}
			-1 & 0 & \cos \varphi & 0 \\
			0 & -1 & \sin \varphi & 0 \\
			\cos\varphi & \sin\varphi & I & 0 \\
			0 & 0 & 0 & J
		\end{matrix}
		\right).
	\end{equation*}
	Hence, the trajectories of the Euler-Lagrange equations for $L^{mod}$ are just the geodesics with respect to $g^{mod}$. Denote by $\text{exp}^{mod}_{q}:T_{q}Q\rightarrow Q$ the exponential map at $q$ associated to $g^{mod}$ and by $i_{q}:\D_{q}\hookrightarrow T_{q}Q$ the inclusion map. The result obtained in \citep*{BLOCH2009225} may be translated into the equation
	\begin{equation*}
		\text{exp}^{mod}_{q} \circ i_{q}=	\text{exp}^{nh}_{q},
	\end{equation*}
	where $\text{exp}^{nh}_{q}:\D_{q}\rightarrow Q$ is the nonholonomic exponential map. Define now the Riemannian metric on $\D_{q}$ as
	\begin{equation*}
		\mathcal G_{0}=(\text{exp}^{nh}_{q})^{*} g^{mod}.
	\end{equation*}
	The nonholonomic exponential map may be computed to be
	\begin{equation*}
		\text{exp}^{nh}_{q}(u,v)=\left(\frac{u}{v}\sin v, \frac{u}{v}(1-\cos v),u,v\right)
	\end{equation*}
	if $v\neq 0$ and
	\begin{equation*}
	\text{exp}^{nh}_{q}(u,0)=\left(u,0,u,0\right), \quad \text{if} \ v=0.
	\end{equation*}
	Hence, we obtain
	\begin{equation*}
		\mathcal G_{0}=E du \otimes du + F du\otimes dv + F dv \otimes du + G dv \otimes d v,
	\end{equation*}
	with
	\begin{equation*}
	\begin{split}
	& E=\frac{2Iv^2-4+2\sin(v) v+4 \cos(v)}{v^2}, \quad F= \frac{u(v^2+4-3 \sin(v) v-4 \cos(v))}{v^3}\\
	& G = \frac{4 \cos(v) u^2+4 \sin(v) u^2 v-(2 v^2+4) u^2+2 J v^4}{v^4},
	\end{split}
	\end{equation*}
	if $v\neq 0$ and
	\begin{equation*}
		\mathcal G_{0}=I du \otimes du + J dv \otimes d v, \quad \text{if} \ v=0.
	\end{equation*}
	
	This metric is easily seen to satisfy the Gauss condition and moreover, the nonholonomic metric $g^{nh}_{q}$ turns out to be simply the pullback of $g^{mod}$ to the submanifold $\M_{q}^{nh}$!
\end{example}
 
\section{Conclusions and future work}\label{conclusions}	 Given a kinetic nonholonomic system, with configuration space $Q$, we characterize all the Riemannian metrics on the image ${\mathcal M}^{nh}_q$ of the nonholonomic exponential map at a fixed point $q \in Q$ which satisfy the following condition: the minimizing geodesics of these metrics, for sufficiently small times, with starting point $q$ are just the nonholonomic trajectories with the same starting point $q$. We also prove that such metrics on ${\mathcal M}^{nh}_q$ always exist and we illustrate these facts with several examples.

After these results, a lot of work remains to be done. In fact, our idea is to develop a research program in order to discuss the geometric properties of the nonholonomic trajectories for a kinetic nonholonomic system. Some problems which will be covered in this program are the following ones.
\begin{itemize}
\item
{\bf Geodesic flows and the kinetic nonholonomic flow.} Let $(g, \D)$ be a kinetic nonholonomic system with configuration space $Q$, $\Gamma_{(g, \D)}\in \frak{X}(\D)$ the kinetic nonholonomic flow and $\text{exp}^{\Gamma_{(g, \D)}}: M^{\Gamma_{(g, \D)}}\subseteq \D \to Q \times Q$ the nonholonomic exponential map considered in Section \ref{section2} (see (\ref{nh-exponential-map})). Then, one may prove that there exists an open neighborhood ${\mathcal U}_{0(Q)}$ of the zero section $0(Q)$ in $\D$ such that
\[
(\text{exp}^{\Gamma_{(g, \D)}})_{|{\mathcal U}_{0(Q)}}: {\mathcal U}_{0(Q)} \subseteq \D \to Q \times Q
\]
is an embedding. Denote by ${\mathcal M}^{nh} = \text{exp}^{\Gamma_{(g, \D)}}({\mathcal U}_{0(Q)})$. It is clear that ${\mathcal M}^{nh}$ is an embedded submanifold of $Q \times Q$, 
\[
\text{exp}^{nh}: = (\text{exp}^{\Gamma_{(g, \D)}})_{|{\mathcal U}_{0(Q)}}: {\mathcal U}_{0(Q)} \subseteq \D \to {\mathcal M}^{nh} \subseteq Q \times Q
\]
is a diffeomorphism and the following diagram 
\[
\xymatrix{
	{\mathcal U}_{0(Q)} \subseteq \D  \ar[rr]^{\text{exp}^{nh}}\ar[rd]_{(\tau_D)_{|{\mathcal U}_{0(Q)}}}&& {\mathcal M}^{nh} \subseteq Q \times Q\ar[ld]^{(pr_1)_{|{\mathcal M}^{nh}}}\\\
	&Q&
}
\]
is commutative, where $\tau_\D: \D \to Q$ is the vector bundle projection. Note that
\[
(pr_1)_{|{\mathcal M}^{nh}}^{-1}(q) = \text{exp}^{nh}_q({\mathcal U}_{0(Q)} \cap \D_q) \simeq {\mathcal M}^{nh}_q, \; \; \mbox{ for } q \in Q.
\]
Now, denote by $\Gamma^{nh}_{(g, \D)}$ the nonholonomic flow considered as a vector field on ${\mathcal M}^{nh}$, via the diffeomorphism $\text{exp}^{nh}$. Then, proceeding as in this paper, one may find a family of bundle metrics
\[
g^{nh}: V(pr_1)_{|{\mathcal M}^{nh}} \times_{\mathcal{M}^{nh}} V(pr_1)_{|{\mathcal M}^{nh}} \to \R
\]
on the vertical bundle of the fibration $(pr_1)_{|{\mathcal M}^{nh}}: {\mathcal M}^{nh} \to Q$ such that if $q \in Q$, the minimizing geodesics on ${\mathcal M}^{nh}_q \subseteq {\mathcal M}^{nh}$ with starting point $q$ are, for sufficiently small times, the nonholonomic trajectories with the same starting point $q$. 

Moreover, one may consider the geodesic flow $\Gamma^{g^{nh}}$ associated with one of such metrics $g^{nh}$ as a vector field on $V(pr_1)_{|{\mathcal M}^{nh}}$. Then, it would be interesting to discuss the relation between the vector field $\Gamma^{g^{nh}}$ on $V(pr_1)_{|{\mathcal M}^{nh}}$ and the nonholonomic flow $\Gamma^{nh}_{(g, \D)}$ on ${\mathcal M}^{nh}$.
\item
{\bf Nonholonomic Jacobi fields and applications.} Very recently (see \citep*{AMM2}) we introduced, in a natural way, the notion of a nonholonomic Jacobi field along a nonholonomic trajectory $c: I \to Q$ of a kinetic nonholonomic system. In fact, every nonholonomic Jacobi field $Z$ over $c: I \to Q$ is the infinitesimal variation of a uniparametric family of nonholonomic trajectories with initial trajectory $c$. Thus, if all the nonholonomic trajectories of the family have the same starting point $q\in Q$ (that is, $Z$ is zero at the initial point $q$) then, using Theorem \ref{main-theorem}, we deduce that there exists a Riemannian metric on ${\mathcal M}^{nh}_q$ such that the nonholonomic trajectories are geodesic for this metric and $Z$ is a Riemannian Jacobi field along the geodesic $c$. On the other hand, as we know (see, {for instance,} \citep*{docarmo,O'Neill}),  Riemannian Jacobi fields play an important role in the study of the singularities of the Riemannian exponential map and the minimizing properties of the Riemannian geodesics. So, after the the previous comments, one may pose the following question: is it possible, using the nonholonomic Jacobi fields, to discuss the singularities of the nonholonomic exponential map and the minimizing properties of the nonholonomic trajectories as in the case of Riemannian geometry?

\item
{\bf Kinetic {Lagrangianization} of kinetic nonholonomic systems.} Let $(g, \D)$ be a kinetic nonholonomic system with configuration space $Q$. After the results in this paper, another natural question arises: under what conditions can one get a kinetic {Lagrangianization} of the system $(g, \D)$? In other words, under what conditions does there exist a Riemannian metric $g^{nh}$ on $Q$ such that the kinetic nonholonomic trajectories for the system $(g, \D)$ are just the geodesics of the metric $g^{nh}$ with initial velocity in $\D$? {Note that there are examples of kinetic nonholonomic systems admitting such metrics on the whole configuration space and, despite that, the constraint distribution is still not integrable (see \citep*{BLOCH2009225}; see also Example \ref{vertical-rolling-disk} in this paper). On the other hand, the} main result in this paper, Theorem \ref{main-theorem}, may be considered as the first step in order to give an answer to the {previous hard question. We also remark} that if the system admits a kinetic {Lagrangianization} then, using the Legendre transformation associated with the kinetic Lagrangian system induced by the Riemannian metric $g^{nh}$, one may produce a Hamiltonian formulation of the original nonholonomic system. So, our question is related with a classical problem in nonholonomic mechanics, the so-called {Hamiltonization problem.} This problem discusses whether a nonholonomic system admits a Hamiltonian formulation after reduction by symmetries . In this direction, much work has been done in recent years (see, for instance, \citep*{BaGa,EhKoMoRi,FeJo,GaMa,GaMo,Jo,Ko,VeVe}; see also \citep*{BaYa} and the references therein). 

\item
{\bf Levi-Civita connections of Gauss Riemannian metrics associated with a kinetic nonholonomic system.} For a kinetic nonholonomic system $(g, \D)$ with configuration space $Q$, in \cite{Lew} the author describes the set of affine connections on $Q$
\[
\nabla: \frak{X}(Q) \times \frak{X}(Q) \to \frak{X}(Q)
\]
which satisfy the condition
\[
\nabla_XY = {\mathcal P}(\nabla^g_XY), \; \; \mbox{ for } X\in \frak{X}(Q) \mbox{ and } Y \in \Gamma(\D),
\]
where $\Gamma(\D)$ is the set of sections of the distribution $\D$, ${\mathcal P}: \frak{X}(Q) \to \Gamma(\D)$ is the orthogonal projector and $\nabla^g$ is the Levi-Civita connection of $g$. The nonholonomic connection $\nabla^{nh}$ considered in Section \ref{section2} (see (\ref{nhconnection})) is a particular example of such affine connections. In fact, in \cite{Lew}, the author proves that the geodesics of any of these connections with initial velocity in $\D$ are just the trajectories of the kinetic nonholonomic system $(g, \D)$. Another important type of affine connections considered in \cite{Lew}, which are related with the system $(g, \D)$, are the so-called energy-preserving. An affine connection $\nabla$ on $Q$ is energy-preserving for the system $(g, \D)$ if the kinetic energy associated with $g$ is constant along the geodesics of $\nabla$. So, it would be interesting to discuss relations between the previous affine connections and the Levi-Civita connections of the Gauss Riemannian metrics $g^{nh}_q$ on the submanifolds ${\mathcal M}^{nh}_q$, with $q \in Q$. We thank A Lewis for getting our attention to the paper \cite{Lew} and for his comments on a draft version of our paper.

\item
{\bf Geometric integrators for kinetic nonholonomic systems.} 
In \cite{klas}, the authors show that several constructions of geometric integrators for nonholonomic mechanics that appear in the literature do not behave well for general nonholonomic systems. Therefore, the problem of finding structure preserving integrators for nonholonomic systems is completely open. However, observe that Theorem \ref{main-theorem} {gives us} the possibility  of considering a new class of  variational type integrators for nonholonomic mechanics (see \cite{MW2001}). 
 For instance, we can consider a  retraction map {$R: TQ\rightarrow Q \times Q$} on a manifold $Q$ (see \cite{absil})  and define the following  discrete {nonholonomic submanifold of $Q \times Q$:}
{$$
R(\D)={\mathcal  M}^{nh,d}\; .$$}
From the properties of retraction maps we have that if {$q \in Q$,} $(R_q)\big|\D_q$ in a neighbourhood of $0_q$ is a diffeomorphism onto its image {$R(\D_q)={\mathcal  M}_q^{nh,d}$.} 
In a future paper, we will explore the construction of variational type integrators on {${\mathcal  M}^{nh,d}$.} One possibility is first to induce, for all $q\in Q$, a Riemannian metric $g^{nh, d}_q$ on ${\mathcal  M}_q^{nh,d}$ as in Theorem \ref{main-theorem}. This metric can be  induced by a Riemannian metric in $\D_{q}$ verifying Gauss condition and using $\left(R\big|_{\D_q}\right)^{-1}$. 
 Then  we can define a discrete Lagrangian $L^{nh}_d: {\mathcal  M}^{nh,d}\rightarrow {\mathbb R}$  as an approximation of the corresponding action: 
 {\begin{eqnarray*}
 	L^{nh}_d(q, q')&\approx&\frac{1}{2}\int^h_0 g^{nh, d}_{q}(c(t))(\dot{c}(t), \dot{c}(t))\; dt
 	\end{eqnarray*}}
where $c: [0,h]\rightarrow {\mathcal  M}_q^{nh,d}$ is the unique geodesic curve for $g^{nh, d}_q$  satisfying {$c(0)=q$ and $c(h)=q'$.} Since  we naturally  have a discrete exact version {(see \citep*{AMM})} then we  could even  study error analysis and backward error analysis for nonholonomic mechanics (see \citep*{MW2001, Hairer}, {for the case of unconstrained Lagrangian systems}). 

\item {\bf Kinetic nonholonomic systems with affine constraints.} It would also be interesting to formulate the analogous results for the special case of kinetic nonholonomic systems with  affine constraints with a moving energy (see \cite{fasso1,fasso2}). The argument would be very similar since, in this case, there exists a change of coordinates that transforms the system into a nonholonomic system with linear constraints where the moving energy is precisely the energy of the transformed system.

\bigskip

Some of the previous problems {on kinetic nonholonomic systems} may be posed for the more general case of nonholonomic Lagrangian systems of mechanical type. Anyway, the first steps, in this direction, should be aimed at trying to extend the results of this paper for such nonholonomic systems as we will explain in the following item.
\item
{\bf Nonholonomic Lagrangian systems of mechanical type versus unconstrained Lagrangian systems of the same type.} The Lagrangian function $L: TQ \to \mathbb{R}$ of an unconstrained mechanical system is given by
\[
L_{(g, V)}(u_q) = \frac{1}{2} g(q) (u_q, u_q) - V (q), \; \; \mbox{ for } u_q \in T_qQ,
\]
where $g$ is a Riemannian metric on $Q$ and $V: Q \to \R$ is the potential energy. In the presence of a constraint distribution $\D$ on $Q$, we have a nonholonomic Lagrangian system $(L_{(g, V)}, \D)$ of mechanical type. In this paper, we deal with kinetic nonholonomic systems, that is, we assume that the potential forces are not present (in other words, $V = 0$). So, a natural question arises: does there exist an unconstrained Lagrangian system of mechanical type such that the nonholonomic trajectories of the system $(L_{(g, V)}, \D)$ with a fixed starting point $q\in Q$ are the trajectories of the unconstrained Lagrangian system with the same starting point $q$ and initial velocity in $\D$? 

%

\end{itemize}


\section*{Acknowledgments} 
D. Mart{\'\i}n de Diego and A. Simoes  acknowledge financial support from the Spanish Ministry of Science and Innovation, under grants PID2019-106715GB-C21, MTM2016-76702-P, ``Severo Ochoa Programme for Centres of Excellence in R\&D'' (SEV-2015-0554) and from the Spanish National Research Council, through the ``Ayuda extraordinaria a Centros de Excelencia Severo Ochoa'' (20205\-CEX001). A. Simoes is supported by the FCT (Portugal) research fellowship SFRH/BD/129882/2017. J.C. Marrero  acknowledges the partial support by European Union (Feder) grant PGC2018-098265-B-C32.

\bibliography{thesisreferences}{}

\end{document}